\documentclass[11pt,a4paper]{article}
\usepackage{indentfirst}
\usepackage[latin1]{inputenc}
\usepackage{mathtools}
\usepackage{amsfonts}
\usepackage[english]{babel}
\usepackage{amssymb}
\usepackage{graphicx}
\usepackage{geometry}
\usepackage{amsthm}
\usepackage{amsmath}
\usepackage{bm}
\usepackage{breqn}

\newtheorem{theorem}{Theorem}[section]

\usepackage{titlesec}
\titlelabel{\thetitle.\quad}
\usepackage{titlesec}

\AtBeginDocument{%
}

\begin{document}
\title{Modified General Relativity and Quantum Theory in curved spacetime}
\author{Gary Nash
	\\University of Alberta, Edmonton, Alberta, Canada,T6G 2R3\\gnash@ualberta.net \footnote{Present address, Edmonton, Alberta}
	}
\date{~November 23, 2021}
\maketitle
\vspace{2mm}
\begin{abstract}With appropriate modifications, the multi-spin Klein-Gordon (KG) equation of quantum field theory can be adapted to curved spacetime for spins 0,1,1/2. The associated particles in the microworld then move as a wave at all spacetime coordinates. From the existence in a Lorentzian spacetime of a line element field $(X^{\beta},-X^{\beta}) $, the spin-1 KG equation $\nabla_{\mu}\nabla^{\mu}X^{\beta}=k^{2}X^{\beta} $ is derived from an action functional involving $X^{\beta}  $ and its covariant derivative. The spin-0 KG equation and the KG equation of the outer product of a spin-1/2 Dirac spinor and its Hermitian conjugate are then constructed. Thus, $ X^{\beta} $ acts as a fundamental quantum vector field. The symmetric part of the spin-1 KG equation, $ \tilde{\varPsi}_{\alpha\beta}$, is the Lie derivative of the metric. That links the multi-spin Klein-Gordon equation to Modified General Relativity (MGR) through its energy-momentum tensor of the gravitational field. From the invariance of the action functionals under the diffeomorphism group Diff(M), which is not restricted to the Lorentz group, $ \tilde{\varPsi}_{\alpha\beta}$ can instantaneously transmit information along $ X^{\beta} $. That establishes the concept of entanglement within a Lorentzian formalism.  The respective local/nonlocal characteristics of MGR and quantum theory no longer present an insurmountable problem to unify the theories.
	
\end{abstract}
%\begin{keyword}
Keywords\\
	{quantum mechanics; Klein-Gordon; general relativity; gravitational energy-momentum; quantum field theory; quantum gravity}
%\ccode{PACS numbers:}
\vspace{2mm}
\newline Preprint of an article published in International Journal of Modern Physics A, Vol. 36, No. 29 (2021) 2150196 (26 pages) DOI: 10.1142/S0217751X21501967 ©copyright World Scientific Publishing Company https://www.worldscientific.com/worldscinet/ijmpa
\section{Introduction}
It has been nearly a century since Schr{\"o}dinger \cite{1} wrote down his equation describing non-relativistic quantum mechanics. In the same year of 1926, Klein \cite{2}, Gordon \cite{3} and Fock \cite{4} developed the relativistic quantum mechanical wave equation; mainly referred to as the Klein-Gordon (KG) equation, which was fundamental to the ensuing development of quantum field theory. Over a decade before that, Einstein \cite{5} developed General Relativity (GR) in 1915. And yet today, despite many subsequent successful developments in both quantum field theory and general relativity, there is still not a full understanding of the relationship between the two fundamental theories of physics. The quantization of gravity has been the major approach to unite the two theories, with string theory and loop quantum gravity the two mainstream proposals. However, those and other theories of quantum gravity have well-documented \cite{6,7} successes and failures. The opposite approach of gravitizing quantum mechanics attempts to bring quantum theory in line with the principles of general relativity as discussed in \cite{8,9} (and other references therein).\par  Rather than forcing quantum theory on general relativity, or vice versa, this article investigates if a natural unification of the theories can be established; unification in the sense that each theory contains the same entity, which links the theories together. Since GR is a local theory and quantum theory has nonlocal characteristics, a symmetric tensor that is common to both theories and solves the local/nonlocal issues must be introduced to achieve the goal of unifying them. That unification is sought from an extended theory of GR, and the covariant wave equations for a spin-1 boson and the parent spin-1/2 equation for a Dirac spinor, obtained from modifications of the Proca and Dirac equations. It is therefore imperative that the proper wave equation of each quantum field in curved spacetime is established before exploring the unification of GR with quantum theory, or proceeding with the quantization of that field. 
\par The relativistic KG wave equation for the quantum field $\Psi$ representing a free field with a particular spin in Minkowski spacetime is fundamental to the formulation of quantum field theory (QFT). In particular, the Klein-Gordon operator $ \partial_{\mu}\partial^{\mu}-k^{2} $, where $ k=\frac{m_{0}c}{\hbar} $ and $ m_{0} $ is the rest mass attributed to each particle of a given spin, plays a fundamental role \cite{10} in the quantization of the multi-spin quantum field $\Psi$. Similarly, a covariant wave equation for $ \Psi $ in curved spacetime is postulated. With appropriate modifications, it will be shown that the covariant wave equation
\begin{equation}\label{KG}
	\nabla_{\mu}\nabla^{\mu}\Psi=k^{2}\Psi, 
\end{equation}explicitly describes spins 0,1,2,1/2 in curved spacetime. The connection is the torsionless Levi-Civita connection for tensor fields, and for spinors, it contains the spinor affinities. For consistency with quantum field theory, we will refer to that equation as the KG equation, but with the understanding that it is the rudimentary covariant wave equation in curved spacetime with the field $ \Psi $ considered to represent spins 0,1,2,1/2. A particle in the quantum environment of the microworld described by that KG equation moves as a wave at all spacetime coordinates.\par Canonical quantization in curved spacetime is discussed in many references such as \cite{11,12}. In curved spacetime, there is no unique way of defining the positive and negative frequency modes and therefore no unique vacuum state. The concept of old-new states is introduced and if mixing of positive and negative frequency solutions occurs, then particles are created by the gravitational field. However, the concepts of vacuum and particles lose their meaning in curved spacetime. As stated in \cite{13}, although natural notions of vacuum state and particles can be defined for a free field in stationary spacetimes, no such notions exist in a general curved spacetime. It is not that particles cannot be defined at all in curved spacetime but rather that many definitions exist and none appears preferred. This difficulty is not present in the algebraic approach. Algebraic QFT \cite{13,14,15} can be formulated without requiring a preferred representation of the canonical commutation relations, and without requiring the definition of a preferred notion of particles. \par Canonical, algebraic and other quantization techniques apply to the asymmetric wave equation (\ref{KG}). However, it is emphasized that the proper wave equation for each quantum field must be established before proceeding with the quantization of that field. Thus, this article is concerned with the construction of the fundamental KG equation (\ref{KG}) for each of the stated spins and not with its subsequent quantization.
\par Curved spacetime is described on a time-oriented Lorentzian manifold with metric $(M,g)$.  On every non-compact paracompact Hausdorff manifold M a vector field $ X^{\beta} $ exists \cite{16} from which the smooth regular line element field $(X^{\beta},-X^{\beta}) $ is defined; it is like a vector field with an undetermined sign \cite{17}. \par The spin-1 KG equation of QFT can be derived from an action functional involving $ X^{\beta} $ and the covariant derivative of its covector $ \varPsi_{\alpha\beta}=\nabla_{\alpha}X_{\beta} $. The spin-0 KG equation is constructed from the scalar $ \Phi=\varPsi_{\alpha\beta}g^{\alpha\beta} $. It will be shown that the KG equation of the outer product of the spin-1/2 Dirac spinor and its Hermitian conjugate yields twice the spin-1 KG equation. Thus, the line element vector field plays the role of a fundamental quantum vector field. Quantum field theory for the noted spins is to be interpreted from the KG equation (\ref{KG}) expressed in terms of the line element field. It is the line element field that provides the extra freedom in \emph{both} MGR and quantum theory to enable the unification of the theories. \par  The left side of the spin-1 KG wave equation (\ref{KG}) can then be symmetrized. As GR involves symmetric tensors, the only possibility to associate GR directly with quantum theory in curved spacetime is through the symmetric part of the KG equation. In that regard, it is noteworthy that the Lie derivative of the metric along $ X^{\beta} $ \emph{is} the symmetric part of the spin-1 KG equation, $ \tilde{\varPsi}_{\alpha\beta} $. That immediately exhibits a geometrical property of quantum theory that has been neglected; quantum theory is incomplete. 
\par Einstein developed GR in a four-dimensional Riemannian spacetime with the understanding that spacetime is locally Minkowskian during free-fall. Recently, the natural extension of GR using the line element fields was developed in the article: Modified General Relativity (MGR) \cite{18}. By adopting Einstein's original postulate of a total energy-momentum tensor that includes the energy-momentum of the gravitational field, and using the Orthogonal Decomposition Theorem (Appendix A), a symmetric tensor $ \varPhi_{\alpha\beta} $ is introduced in MGR that describes the energy-momentum of the gravitational field and \emph{completes} GR. $ \varPhi_{\alpha\beta} $ is constructed from the Lie derivative along the line element vector of both the metric and a product of unit line element covectors. Hence, $ \tilde{\varPsi}_{\alpha\beta} $ is contained in both $ \varPhi_{\alpha\beta} $ and the KG equation (\ref{KG}) for spins 1,0 and an outer product of spin-1/2 Dirac spinors. The Lie derivative of the metric along the line element field connects MGR and quantum theory in curved spacetime. Moreover, it will be established that the Einstein equation intrinsically resides in the symmetric part of the spin-2 KG equation. \par The action functional for MGR and the KG action functional for spins 0,1 are invariant under the diffeomorphism group, Diff(M). $ \tilde{\varPsi}_{\alpha\beta} $ is constructed from the pullback of the metric under Diff(M), which is not restricted to the Lorentz group. Hence, the rate of change of the gravitational field along the regular vector $ X^{\beta} $ is not limited to the speed of light. The metric at a point on the integral curve associated with $ X^{\beta} $, far from a given point $ p $ on that curve, can be pulled back instantaneously to the neighbourhood of $ p $, or pushed forward from $ p $. The Lie derivative of the metric along the line element field vector, being the symmetric part of the spin-1 KG wave equation, opens the door to a discussion of the concept of entanglement within a Lorentz invariant formalism. Thus, it is important to develop the relationship of $ \tilde{\varPsi}_{\alpha\beta} $ to the spin-0 KG equation and the KG wave equation for a pair of Dirac spinors. That will expose the physical property of entanglement for spin-0 and spin-1 bosons, and a pair of spin-1/2 fermions.\par Some interesting results appear from the study of the spin-2 KG equation. Gravity, described by the metric, is a long-range \emph{effective} force. If gravitons are the exchange particles of gravity, they must be massless. In a four-dimensional time-oriented Lorentzian spacetime, it is shown that massless gravitons governed by the spin-2 KG equation cannot be described with the metric, and massless spin-2 ``particles" do not couple to a non-zero energy-momentum tensor. Therefore, massless gravitons do not act as force mediators of gravity. This result should not be considered controversial. It is well known that GR, as a classical field theory, does not require particle exchange to describe the effective force of gravity; that is nicely done by the curvature of spacetime. Moreover, there is nothing in the formalism of QFT that requires GR to be quantized; that was noted by Feynman, who said \cite{19} ``It is still possible that quantum theory does not absolutely guarantee that gravity \emph{has} to be quantized."  
\par The covariant multi-spin KG equation (\ref{KG}) is a pure wave equation in a curved spacetime quantum environment. Despite the fact that the process of measurement may constrain the wave behaviour of a particle, it does not permanently destroy that property of the particle when detected as such, and it does not suddenly become a wave when detected as a wave, because in both cases, the particle in the quantum environment of the microworld fundamentally moves as a wave at all spacetime coordinates in accordance with the covariant KG equation (\ref{KG}). That removes the apparent mystery in the wave-particle duality, which is a fundamental part of quantum theory. The apparent mystery is due to the physical process of measuring a particle in the microworld, which requires the statistical approach of quantum theory to describe the chaotic event. However, there is no wave-particle mystery in the microworld; the particles move as a wave at all possible spacetime coordinates.\par 
The KG equation for a spin-1 boson or a specific pair of spin-1/2 fermions, contains the symmetric tensor $ \tilde{\varPsi}_{\alpha\beta} $ in addition to the antisymmetric tensor $ K_{\alpha\beta} $, which traditionally describes the particles. $ K_{\alpha\beta} $ satisfies a wave equation that contains derivatives of the Lie derivative of the metric. The noted particles move as a wave at all spacetime coordinates and are guided by local changes in the gravitational field, which are significant at small distances. That property is similar in concept to the pilot wave of the de Broglie-Bohm theory \cite{20,21}, but it does not involve hidden variables and the mysterious guiding equation, which describes the pilot wave. \par Einstein declared quantum theory to be incomplete \cite{22} and introduced hidden variables into the theory in order to restore the concept of locality. However, Bell \cite{23} proved that it is impossible to construct a hidden variable theory that obeys locality and simultaneously reproduces all of the properties of quantum mechanics. Quantum theory is intrinsically nonlocal and hidden variables cannot be introduced to restore locality. \par  Nonlocality is a natural property of the symmetric part of the multi-spin KG equation (\ref{KG}). It is possible for a specific pair of spin-1/2 fermions to transmit quantum information instantaneously along the integral curve associated with the spin-1 quantum vector from which the Lie derivative of the metric is constructed. That establishes the concept of entanglement and the global nature of quantum theory. Thus, by using the line element field within MGR and introducing a symmetric tensor into quantum theory by symmetrizing the KG equation in curved spacetime, the respective local/nonlocal characteristics of MGR and quantum theory no longer present an insurmountable problem to unify the theories.  \par The structure of the paper is organized as follows: a partial review of MGR is presented in section 2 and section 3 explores the multi-spin covariant Klein-Gordon wave equation for spins 0,1,2,1/2. In section 4, spin-2 ``particles" are discussed, and the unification of the spin-1 fields with the Lorentzian metric is shown in section 5. The wave-particle duality, determinism and quantum entanglement are discussed in section 6, followed by the conclusion of section 7.
\section{Modified General Relativity}
Curved spacetime is described by the four-dimensional time-oriented Lorentzian manifold with metric, $ (M,g_{\alpha\beta}) $. The manifold describing spacetime in this paper is considered to be Hausdorff, connected, non-compact paracompact and smooth without boundary. These properties of the manifold lead to the structure of a Lorentzian metric. A smooth regular vector field $ X^{\alpha} $ exists on any non-compact paracompact Hausdorff manifold M  \cite{16,24}, from which the line element field $(\bm{X},-\bm{X}) $ is defined as an assignment of a pair of equal and opposite vectors at each point of M. It is like a vector field with an undetermined sign \cite{17}. A non-compact paracompact manifold admits a Lorentzian metric $ g_{\alpha\beta} $ if and only if it admits a line element field \cite{17}. A Lorentzian metric consists of a Riemannian metric $ g_{\alpha\beta}^{+} $, which always exists on $M$, and a unit covector $ u_{\alpha} $ that is collinear with one of the pair of vectors in the line element field \cite{16,17,24,25}:
\begin{equation}\label{gab}
	g_{\alpha\beta}=g_{\alpha\beta}^{+}-2u_{\alpha}u_{\beta},
\end{equation}which means the inverse metric is $ g^{\alpha\beta}=g^{+\alpha\beta}-2u^{\alpha}u^{\beta} $ with $ u_{\alpha}u^{\alpha}=1 $ using $ g^{+} $ and $u_{\alpha}u^{\alpha}=-1   $ using $ g $. The Riemannian metric $ g^{+} $ is independent of $ u $, which demands the Lorentzian metric $ g $ to be independent of $ X $; and conversely. \par Thus, a Lorentzian spacetime does not exist without a line element field. It is not enough to say a vector field exists on $ M $; a Lorentzian metric must have a directional attribute, which is what the line element field provides. MGR exploits one of the pair of vectors in the line element field to complete GR without introducing any additional vectors on the Lorentzian manifold. \par On any non-compact paracompact Hausdorff manifold there exists a Lorentz structure that does not generate closed timelike curves \cite{16}. The spacetime is assumed to admit a Cauchy surface and is therefore globally hyperbolic. That forbids the presence of closed causal curves \cite{25}. The connection on the manifold is torsionless and metric compatible.\par 
For completeness and reference, a partial review of MGR \cite{18} is now discussed. In MGR, Einstein's original postulate of a total energy-momentum
tensor $ T_{\alpha\beta} $ is adopted. It contains the matter energy-momentum tensor $ \tilde{T_{\alpha\beta}} $ and, since gravity gravitates, must admit a symmetric tensor that geometrically describes a dynamically changing metric to represent the energy-momentum of the gravitational field. That demands the local conservation of total energy-momentum with the vanishing of $\nabla^{\alpha}T_{\alpha\beta} $ and not $\nabla^{\alpha}\tilde{T}_{\alpha\beta}  $. Since $ \tilde{T}_{\alpha\beta} $ is not divergenceless, it can be set proportional to an arbitrary symmetric tensor $ w_{\alpha\beta} $, which is then orthogonally decomposed by the Orthogonal Decomposition Theorem (ODT); a revised version of which is included as Appendix A for reference in developing the spin-2 field equation. In addition to the Lovelock tensors, this introduces a new tensor, $ \varPhi_{\alpha\beta} $, which represents the energy-momentum of the gravitational field. It is constructed from the Lie derivatives along $ X^{\beta} $ of both the metric and a product of unit line element covectors:
\begin{equation}\label{phiab}
	\begin{split}
		\varPhi_{\alpha\beta}&=\frac{1}{2}\pounds_{X}g_{\alpha\beta}+\pounds_{X}u_{\alpha}u_{\beta}\\
		&=\frac{1}{2}(\nabla_{\alpha}X_{\beta}+\nabla_{\beta}X_{\alpha})+u^{\lambda}(u_{\alpha}\nabla_{\beta}X_{\lambda}+u_{\beta}\nabla_{\alpha}X_{\lambda}).
	\end{split}
\end{equation}
Since Lie derivatives are independent of the connection, they can be expressed in terms of covariant or partial derivatives. $\varPhi_{\alpha\beta}$ does not vanish when the connection coefficients vanish. Hence, the metric can be locally Minkowskian, as in free fall, without forcing $\varPhi_{\alpha\beta}$ to vanish. Local Lorentz invariance is preserved and there is no conflict with the equivalence principle. The radially symmetric static energy density of the gravitational field is calculated to be $-\frac{c^{4}}{8\pi G}\varPhi_{00}=-\frac{GM^{2}}{4\pi r^{4}}-\frac{c^{4}a_{0}^{2}}{16\pi G} $, which, since gravity gravitates,  is twice the Newtonian result and includes a contribution from the dark energy parameter $ a_{0} $. Gravitational energy-momentum is invariant under free-fall and can be localized. The total energy-momentum tensor $T_{\alpha\beta}=\tilde{T}_{\alpha\beta}-\frac{c^{4}}{8\pi G}\varPhi_{\alpha\beta}$ is divergenceless by virtue of the diffeomorphic invariance of the theory, and is therefore locally conserved.  
\par The action functional $ S=S^{F}+S^{EH}+S^{G} $ consists of the action for all matter fields $S^{F}$, the Einstein-Hilbert action of general relativity $S^{EH}$ and the action for the energy-momentum of the gravitational field, $S^{G}=-b\int\varPhi_{\alpha\beta} g^{\alpha\beta}\sqrt{-g}d^{4}x$:
\begin{equation}\label{Se}
	\begin{split}
		S=\int L^{F}( A^{\beta},\nabla^{\alpha} A^{\beta},...,g^{\alpha\beta})\sqrt{-g}d^{4}x
		+b \int (R-2\Lambda)\sqrt{-g}d^{4}x\\-b\int \varPhi_{\alpha\beta} g^{\alpha\beta}\sqrt{-g} d^{4}x
	\end{split}
\end{equation}where $L^{F}  $ is the Lagrangian for the matter fields. The variation of the action functional $S$ with respect to the inverse metric $g^{\alpha\beta}$ is
\begin{multline}
	\delta S =\int [-\frac{1}{2c}\tilde{T}_{\alpha\beta} +b(R_{\alpha\beta}-\frac{1}{2}g_{\alpha\beta}R)+b\Lambda g_{\alpha\beta}+b(\nabla_{\alpha}X_{\beta}
	+2u^{\lambda}u_{\beta}\nabla_{\alpha}X_{\lambda}\\
	+\nabla_{\mu}X_\nu(-u_{\alpha}u_{\beta} g^{\mu\nu}+u^{\mu}u^{\nu}g_{\alpha\beta})
	)]\delta g^{\alpha\beta}
	\sqrt{-g}\enspace d^{4}x\enspace+2b\int \nabla_{\alpha}(u^{\alpha}u^{\beta})\delta X_{\beta}\sqrt{-g}d^{4}x.
\end{multline} With $\delta S=0$ and arbitrary variations for $\delta g^{\alpha\beta}$ and $\delta X_{\beta}  $, it follows with $b=\frac{c^{3}}{16\pi G}$ that
\begin{equation}\label{MEQ}
	-\frac{8\pi G}{c^{4}}\tilde{T}_{\alpha\beta}+G_{\alpha\beta}+{\Lambda} g_{\alpha\beta}+\varPhi_{\alpha\beta}=0
\end{equation} is the modified Einstein equation where $ G_{\alpha\beta} $ is the Einstein tensor, and
\begin{equation}\label{nuab}
	\nabla_{\alpha}(u^{\alpha}u^{\beta})=0.
\end{equation}
The trace of $ \varPhi^{\alpha\beta} $ with respect to the metric, which represents the interaction of the gravitational field with its energy-momentum tensor, defines the scalar $\Phi=g_{\alpha\beta}\varPhi^{\alpha\beta}$. Using (\ref{nuab}), its integral over all spacetime vanishes:
\begin{equation}\label{intPhi}
	\int \Phi \sqrt{-g}d^{4}x=0.
\end{equation}Although the integral $ S^{G} $ vanishes, it admits the variation $ \delta S^{G}=b\int\varPhi_{\alpha\beta}\delta g^{\alpha\beta}\sqrt{-g}d^{4}x $ and from (\ref{intPhi}), the action
\begin{equation}\label{EHG}
	S^{EHG}=\frac{c^{3}}{16\pi G}\int (R-\Phi)\sqrt{-g}d^{4}x
\end{equation} generates the modified Einstein equation with no cosmological constant. If $ \Phi$ is set equal to 
$2\Lambda $, the Einstein equation with a cosmological constant is obtained accordingly, which contradicts (\ref{intPhi}). Thus, $ \Phi $ dynamically replaces the cosmological constant. $\varPhi_{\alpha\beta}$ completes the Einstein equation leaving it intact in form:
\begin{equation}\label{ME}
	G_{\alpha\beta}=\frac{8\pi G}{c^{4}}T_{\alpha\beta}.
\end{equation}
\par Varying the action $ S $ with respect to $ X^{\mu} $ determines the dynamical properties of the line element field:
\begin{equation}\label{Xvar}
	\frac{\delta S}{\delta X^{\mu}}=\int [2u^{\alpha}\varPhi_{\alpha\nu}+2u^{\alpha}u^{\lambda}(u_{\alpha}\nabla_{\nu}X_{\lambda}+u_{\nu}\nabla_{\alpha}X_{\lambda})-\Phi u_{\nu}]\frac{\delta u^{\nu}}{\delta X^{\mu}}\sqrt{-g}d^{4}x=0.	
\end{equation} That yields the Lorentz invariant expression
\begin{equation}\label{umu}
	u_{\nu}=\frac{3\partial_{\nu}f }{\Phi} 
\end{equation}
where $ f\neq0 $ is the magnitude of $ X^{\alpha} $ dual to $ X_{\nu}=fu_{\nu} $ by $ g^{\alpha\nu} $. Of the infinite number of possible covectors, only those given by (\ref{umu}) belong to the metric of MGR. It is important to note from (\ref{umu}) and (\ref{gab}) that $ \Phi $ is an intrinsic part of the spacetime metric. \par In general, $ \Phi $ cannot vanish because even weak gravitational fields gravitate. In fact, $\varPhi_{\alpha\beta}  $ vanishes if and only if $ X^{\mu} $ is a Killing vector. In an affine parameterization,  $\Phi=\frac{1}{2}(\nabla_{\alpha}X_{\beta}+\nabla_{\beta}X_{\alpha})g^{\alpha\beta} $, which vanishes if $ X^{\mu} $ is a Killing vector. Consequently, $\varPhi_{\alpha\beta}  $ must vanish because the inverse metric is non-degenerate, and conversely. However, in general, there are no Killing vector fields unless a particular symmetry is involved.
\section{The multi-spin covariant Klein-Gordon equation}In Minkowski spacetime, there exists a relativistic multi-spin Klein-Gordon equation for free fields \cite{10}. The commutation of partial derivatives allows the spin-1 field to be described by the KG wave equation, and the KG wave equation for a Dirac spinor is the parent equation of the two Dirac equations. In curved spacetime, covariant derivatives do not commute. The Proca spin-1 equation is no longer a pure wave equation, and the product of the two operators in the Dirac equations acting on a Dirac spinor is not equivalent to the wave equation for a Dirac spinor. However, with appropriate modifications, it is possible to develop the covariant wave equation (\ref{KG}) for spins 0,1,1/2 that is crucial for the unification of quantum field theory and MGR.
\par  To achieve that goal, the (0,2) tensor
\begin{equation}\label{Psi}
	\varPsi_{\alpha\beta}:=\nabla_{\alpha}X_{\beta}
\end{equation}is required. It can be symmetrized according to
\begin{equation}\label{Sym}
	\begin{aligned}
		\varPsi_{\alpha\beta} &=\frac{1}{2}(\nabla_{\alpha}X_{\beta}
		+\nabla_{\beta}X_{\alpha})+\frac{1}{2}(\nabla_{\alpha}X_{\beta} -\nabla_{\beta}X_{\alpha})\\
		&:=\frac{1}{2}\tilde{\varPsi}_{\alpha\beta}+\frac{1}{2}K_{\alpha\beta}.
	\end{aligned}
\end{equation} There are now some subtle points to discuss. First, the symmetric tensor $\tilde{\varPsi}_{\alpha\beta}$ is the Lie derivative of the metric along $ X^{\beta}$: $ \pounds_{X}g_{\alpha\beta}=\nabla_{\alpha}X_{\beta}+\nabla_{\beta}X_{\alpha} $. Second, the Lorentz constraint
\begin{equation}\label{LC}
	\nabla_{\alpha}X^{\alpha}=0
\end{equation}cannot be invoked for spins 1,0 because it would force $\tilde{\varPsi}_{\alpha\beta}  $ to vanish which  cannot be true in general. Even weak gravitational fields gravitate and $X^{\beta}  $ is not a Killing vector. It follows from (\ref{phiab}) that $ \Phi $ is given by
\begin{equation}\label{PhiNX}
	\Phi=\nabla_{\alpha}X^{\alpha}\neq0
\end{equation} in an affine parameterization where the geodesic term  $ 2u_{\lambda}u^{\alpha}\nabla_{\alpha}X^{\lambda} $ vanishes. $ \Phi $ replaces the Lorentz constraint. Third, the additional divergenceless constraint
\begin{equation}\label{DPsi}
	\nabla_{\alpha}\tilde{\varPsi}^{\alpha\beta}=0
\end{equation} cannot apply to the KG spin-1 and spin-1/2 equations as shown below.\par The spin-1 KG equation is now derived directly from an action functional involving $ X^{\beta} $ and its covariant derivative. That provides the fundamental relationship of the line element field to the spin-1 KG equation and supports why $X^{\beta}  $ plays the role of a fundamental quantum vector field.
\subsection{Derivation of the spin-1 Klein-Gordon equation from the line element field}
The action functional is
\begin{equation}\label{key}
	S^{1}=-\frac{1}{2}\int[\nabla_{\sigma}X_{\varrho}\nabla^{\sigma}X^{\varrho}+k^{2}X_{\sigma}X^{\sigma}]\sqrt{-g}d^{4}x.
\end{equation}Using (\ref{gab}), the variation of $ S^{1} $ with respect to $ X^{\nu} $ yields
\begin{equation}\label{deltaS}
	\begin{split}
		\int[\Box X_{\nu}-k^{2}X_{\nu}+[k^{2}(X^{\sigma}X_{\sigma}u^{\kappa}-2X^{\sigma}u_{\sigma}X^{\kappa})+X^{\varrho}u_{\varrho}\Box X^{\kappa}+X^{\kappa}u_{\beta}\Box X^{\beta}-X_{\varrho}u^{\kappa}\Box X^{\varrho}]\\\frac{\delta u_{\kappa}}{\delta X_{\nu}}]\delta X_{\nu}\sqrt{-g}d^{4}x
		+\int[u^{\kappa}\nabla_{\sigma}(X_{\varrho}\nabla^{\sigma}X^{\varrho})+2u_{\varrho}(\nabla^{\alpha}X^{\varrho})\nabla_{\alpha}X^{\kappa}]\frac{\delta u_{\kappa}}{\delta X_{\nu}}\delta X_{\nu}\sqrt{-g}d^{4}x\\
		-\frac{1}{2}\int\nabla_{\varrho}(X^{\beta}\nabla^{\alpha}X^{\varrho}-X^{\varrho}\nabla^{\alpha}X^{\beta})\frac{\delta g_{\alpha\beta}}{\delta u_{\kappa}}\frac{\delta u_{\kappa}}{\delta X_{\nu}}\delta X_{\nu}\sqrt{-g}d^{4}x
	\end{split}
\end{equation}after calculating $ \delta \Gamma^{\lambda}_{\alpha\beta} $ induced by variations of $ u_{\kappa} $ in the metric and defining $ \Box :=\nabla_{\alpha}\nabla^{\alpha} $. The last tensor term of (\ref{deltaS}) vanishes, which follows from choosing a basis $(e_{\alpha})$ at a point $p\in M$ orthonormal with respect to g and with $ e_{0}=u $. Then $ u^{0}=1, u_{0}=-1, u^{j}=u_{j}=0 $ for $ j=1,2,3 $, $g_{\alpha\beta}=\eta_{\alpha\beta}$ and $ X^{0}=fu^{0} $, $X^{j}=0$ where $ f $ is the magnitude of $ X^{\beta} $. After setting the variation of $ S^{1} $ to zero, it follows that the spin-1 KG equation 
\begin{equation}\label{}
	\nabla_{\alpha}\nabla^{\alpha}X^{\nu}=k^{2}X^{\nu}
\end{equation} is a solution to the variational equation obtained from (\ref{deltaS}) provided the line element fields satisfy the constraint
\begin{equation}\label{Xcons}
	u^{\kappa}\nabla_{\alpha}(X_{\varrho}\nabla^{\alpha}X^{\varrho})+2u_{\varrho}(\nabla^{\alpha}X^{\varrho})\nabla_{\alpha}X^{\kappa}=0,
\end{equation}which is the inhomogeneous wave equation 
\begin{equation}\label{f}
	\Box f=-\frac{\partial_{\alpha}f}{f}(\partial^{\alpha}f+2\partial^{\alpha}f_{(\kappa)})
\end{equation}
where $f\neq0 $ is the magnitude of vector $ X^{\varrho} $ and $ f_{(\kappa)}\neq0 $ the magnitude of $ X^{\kappa} $.\par The fundamental goal of this article: to find a natural connection between quantum theory and MGR, is now attainable. Once it is established that the Lie derivative of the metric in a Lorentzian spacetime, $\tilde{\varPsi}_{\alpha\beta}$, is a fundamental part of the multi-spin KG equation for a given spin, it directly connects MGR to quantum field theory by the relation
\begin{equation}\label{GRKG}
	\varPhi_{\alpha\beta}=\frac{1}{2}\tilde{\varPsi}_{\alpha\beta}+u^{\lambda}(u_{\alpha}\nabla_{\beta}X_{\lambda}+u_{\beta}\nabla_{\alpha}X_{\lambda}).
\end{equation}\par It is now shown that the Lie derivative of the metric is rudimentary to the multi-spin KG equation for spins 0,1 and an outer product of spin 1/2 spinors, by examining each spin.
\subsection{Spin-0 Klein-Gordon equation}
When $\Psi$ is the scalar field $ \Phi $, the familiar spin-0 KG equation is 
\begin{equation}\label{KG0}
	\nabla_{\mu}\nabla^{\mu}\Phi=k^{2}\Phi.
\end{equation}If $ \Phi $ is constructed from the line element vector field 
\begin{equation}\label{phi}
	\Phi:=\nabla_{\alpha}X_{\beta}g^{\alpha\beta}
\end{equation}it follows that 
\begin{equation}\label{spin0}
	\nabla_{\mu}\nabla^{\mu}\Phi=\frac{1}{2}g^{\alpha\beta}\nabla_{\mu}\nabla^{\mu}\tilde{\varPsi}_{\alpha\beta},
\end{equation} which exhibits the dependence on the Lie derivative of the metric. 
\subsection{Spin-1 KG equation}
The traditional Proca equation \cite{26,27} in curved spacetime
\begin{equation}\label{key}
	\nabla_{\alpha}K^{\alpha\beta}=k^{2}X^{\beta}
\end{equation} with the Lorentz constraint, completely describe a neutral spin-1 boson for $ k\neq0 $, and the photon for $ k=0 $. However, this equation forces the spin-1 wave equation to be expressed as $\nabla_{\alpha}\nabla^{\alpha}X^{\beta}=k^{2}X^{\beta}+\nabla_{\alpha}\nabla^{\beta}X^{\alpha}$, which violates the fundamental postulate that the spin-1 KG equation shall be a pure wave equation of the form (\ref{KG}) when $ \Psi $ is the vector field $ X^{\beta} $:
\begin{equation}\label{KG1}
	\nabla_{\alpha}\nabla^{\alpha}X^{\beta}=k^{2}X^{\beta}.
\end{equation}
\par This problem is automatically resolved with the inclusion of the symmetric part of the KG equation. The equivalent form of (\ref{KG1}) obtained from (\ref{Sym}) is
\begin{equation}\label{P2}
	\nabla_{\alpha}(\tilde{\varPsi}^{\alpha\beta}+K^{\alpha\beta})=2k^{2}X^{\beta}.
\end{equation} If $\tilde{\varPsi}^{\alpha\beta}  $ is divergenceless, then (\ref{P2}) generates the Lorentz constraint (\ref{LC}) 
\begin{equation}\label{key}
	\begin{split}
		\nabla_{\beta}\nabla_{\alpha}K^{\alpha\beta}=-\frac{1}{2}[\nabla_{\alpha},\nabla_{\beta}]K^{\alpha\beta}=-R_{\alpha\beta}K^{\alpha\beta}=2k^{2}\nabla_{\beta}X^{\beta}=0
	\end{split}
\end{equation}
because $R_{\alpha\beta}  $ and $ K^{\alpha\beta} $ have opposite symmetries. The Lorentz constraint forces $ \tilde{\varPsi}_{\alpha\beta} $ to vanish because $ g^{\alpha\beta} $ is non-degenerate. Thus, (\ref{DPsi}) cannot hold for the spin-1 field. The spin-1 KG equation in curved spacetime contains \emph{both} symmetric and antisymmetric components with the symmetric part being the Lie derivative of the metric. $ \Phi=\nabla_{\alpha}X^{\alpha}\neq0 $ is an intrinsic part of the spin-1 field, which replaces the Lorentz constraint and allows the spin-1 KG equation to be expressed as (\ref{KG1}).
\subsection{Spin-1/2 Dirac spinors and the spin-1 Klein-Gordon equation}  With $ \Psi $ a first rank four-component Dirac spinor $ \varPsi^{a} $, the spin-1/2 KG wave equation in curved spacetime is 
\begin{equation}\label{KG12}
	\nabla_{\alpha}\nabla^{\alpha}\varPsi^{a}=k^{2}\varPsi^{a}
\end{equation}where $a=1,2,3,4$ and the covariant derivative contains the spinor affinities. That is the parent equation of the Dirac equations
\begin{equation}\label{Dirac}
	(\gamma^{\mu}\nabla_{\mu}+k)\varPsi^{a}=0, \enspace (\gamma^{\nu}\nabla_{\nu}-k)\varPsi^{a}=0 
\end{equation}in curved spacetime provided each Dirac equation is modified as shown below. The relationship between Dirac spinors and the Lie derivative of the metric along the line element vector must now be established. To accomplish that, some spinor analysis is presented for clarity. \par 
It is well known \cite{28,29} that the second rank spinors $\varphi^{A\dot{B}}$ and $\varphi_{A\dot{B}}$ can be expressed in terms of the associated vector $X^{\beta}$ and covector $X_{\beta}$ as  \begin{equation}\label{CC}
	\varphi^{A\dot{B}}=\sigma^{A\dot{B}}_{\beta}X^{\beta}
\end{equation} and
\begin{equation}\label{CCI}
	\varphi_{A\dot{B}}=\sigma_{A\dot{B}}^{\beta}X_{\beta}.
\end{equation}
The Hermitian connecting quantities $\sigma^{A\dot{B}}_{\beta}$ transform as a spacetime vector on the index $\beta$ and as spinors on the index $A=1,2$ and conjugate index $\dot{B}=1,2$. Covariant derivatives of spinors are introduced in the same formalism as that for tensors by adopting the spinor affinities $\Gamma^{A}_{\alpha B}$ and defining\begin{equation}
	\nabla_{\alpha}\varPsi_{A}=\partial_{\alpha}\varPsi_{A}-\Gamma^{B}_{\alpha A}\varPsi_{B},\enspace  \nabla_{\alpha}\varPsi^{A}=\partial_{\alpha} \Psi^{A}+\Gamma^{A}_{\alpha B}\varPsi^{B}
\end{equation} for the spinors $\varPsi_{A}$ and $\varPsi^{A}$, respectively.
The covariant derivative of a mixed index spinor-tensor is defined as
\begin{equation}
	\nabla_{\alpha}\varPsi^{\beta A}=\partial_{\alpha}\varPsi^{\beta A}+\Gamma^{\beta}_{\alpha\kappa}\varPsi^{\kappa A}+\Gamma^{A}_{\alpha B}\varPsi^{\beta B}
\end{equation} and the covariant derivative of the connection quantities vanishes
\begin{equation}
	\nabla_{\kappa}\sigma^{\alpha}_{A\dot{B}}=0.
\end{equation}  \par Using these relationships between tensors and spinors, we can now proceed to obtain the Lie derivative of the metric along the line element vector from a KG equation involving Dirac spinors. Let $ \varPsi^{a}_{\dot{a}}:=\varPsi^{a}\otimes\varPsi^{a\dagger} $ represent the second rank spinor formed from the outer product of the Dirac spinor and its Hermitian conjugate (complex transpose). It satisfies the wave equation
\begin{equation}\label{KGC12}
	\nabla_{\alpha}\nabla^{\alpha}\varPsi^{a}_{\dot{a}}={k^{\prime}}^{2}\varPsi^{a}_{\dot{a}}
\end{equation}provided $k^{\prime}=\sqrt{2}k  $ which demands the constraint 
\begin{equation*}\label{key}
	\nabla_{\alpha}\varPsi^{a}\nabla^{\alpha}\varPsi^{a\dagger}=0. 
\end{equation*} $\varPsi^{a}_{\dot{a}} $ is a $4\times4$ matrix which can be represented by four $ 2\times2 $ blocks. The four blocks are second rank spinors identified as 
\begin{equation}\label{Stype}
	\varPsi^{A}_{\dot{A}},\; \varPsi_{A}^{\dot{A}},\; \varPsi^{A\dot{A}}\; \text{and}\; \varPsi_{A\dot{A}} .
\end{equation}
Each of these spinors must satisfy the wave equation. In particular 
\begin{equation}\label{KGB2}
	\nabla_{\alpha}\nabla^{\alpha}\varPsi^{A\dot{A}}={k^{\prime}}^{2}\varPsi^{A\dot{A}}.
\end{equation}Using (\ref{CC}), that is equivalent to
\begin{equation}
	\sigma^{A\dot{A}}_{\beta}\nabla_{\alpha}\nabla^{\alpha}X^{\beta}={k^{\prime}}^{2}\varPsi^{A\dot{A}}
\end{equation} which can be rewritten in terms of $\varPsi^{\alpha\beta}=\nabla^{\alpha}X^{\beta}$ and symmetrized as in (\ref{Sym}) to give\begin{equation}\label{KG122}
	\sigma^{A\dot{A}}_{\beta}\nabla_{\alpha}(\tilde{\varPsi}^{\alpha\beta}+K^{\alpha\beta})=2{k^{\prime}}^{2}\sigma^{A\dot{A}}_{\beta}X^{\beta}.
\end{equation} The three other blocks described in (\ref{Stype}) will lead to (\ref{KG122}) with different connection quantities but with the same vector $ X^{\beta} $. The connection quantities are non-degenerate so
\begin{equation}\label{key}
	\mid\sigma^{A\dot{A}}_{\beta}\mid\mid\nabla_{\alpha}(\tilde{\varPsi}^{\alpha\beta}+K^{\alpha\beta})-2{k^{\prime}}^{2}X^{\beta}\mid=0
\end{equation}requires
\begin{equation}\label{KG1P}
	\nabla_{\alpha}(\tilde{\varPsi}^{\alpha\beta}+K^{\alpha\beta})=2{k^{\prime}}^{2}X^{\beta},
\end{equation}which is the spin-1 equation (\ref{P2}) with the inverse Compton wavelength $ k^{\prime} $. The converse also applies: given (\ref{KG1P}), it follows that $ \varPsi^{a}_{\dot{a}} $ satisfies the wave equation (\ref{KGC12}) with $k^{\prime}=\sqrt{2}k$ and the stated constraint. From the spin-1 situation, the divergenceless condition (\ref{DPsi}) cannot be employed because it would require $\tilde{\varPsi}_{\alpha\beta}  $ to vanish. Thus, the KG equation for $ \varPsi^{a}_{\dot{a}} $ contains the Lie derivative of the metric along $ X^{\beta} $. That two related spin-1/2 spinors must encompass the Lie derivative of the metric is important in the discussion of entanglement in section 6.  \par Further study of the spin-1/2 KG equation reveals that the Ricci scalar is problematic; the product of the Dirac factorizations contains a $ \frac{R}{4} $ term that prevents the parent spin-1/2 KG equation from being a pure wave equation of the form (\ref{KG}). That requires an adjustment to the Dirac equations in curved spacetime so that the product of their operators yields the spin-1/2 KG wave equation. 
\par The gamma matrices $\gamma^{\mu}$ in curved spacetime are assumed to satisfy the anticommutation relation
\begin{equation}\label{gamma}
	\{\gamma^{\mu},\gamma^{\nu}\}=2g^{\mu\nu}.
\end{equation}  Using 
\begin{equation}
	\nabla_{\mu}\gamma^{\nu}=0,
\end{equation} the product of the Dirac factorizations 
\begin{equation}
	(\gamma^{\mu}\nabla_{\mu}+k)(\gamma^{\nu}\nabla_{\nu}-k)\varPsi^{a}=0
\end{equation}yields
\begin{equation}\label{gamma2}
	\gamma^{\mu}\gamma^{\nu}\nabla_{\mu}\nabla_{\nu}\varPsi^{a}=k^{2}\varPsi^{a},
\end{equation}which is expressed in the literature \cite{11,30} (restated here using a +2 signature metric)  as
\begin{equation}
	\nabla_{\mu}\nabla^{\mu}\varPsi^{a}=(k^{2}+\frac{1}{4}R)\varPsi^{a}. 
\end{equation}With only the algebra of (\ref{gamma}), the spin-$1/2$ KG equation in curved spacetime is not precisely recoverable due to the additional $\frac{1}{4}R$ term. That term can be eliminated by inserting the complex scalar $ \Omega $ into (\ref{Dirac}) to obtain the modified Dirac equations 
\begin{equation}\label{key}
	(\gamma^{\mu}\nabla_{\mu}+\Omega+k)\varPsi^{a}=0, \enspace (\gamma^{\nu}\nabla_{\nu}+\Omega-k)\varPsi^{a}=0 
\end{equation}where 
\begin{equation}\label{key}
	\Omega^{2}=\frac{R}{4}.
\end{equation}
By defining the algebra
\begin{equation}\label{key}
	\{\gamma^{\mu}\nabla_{\mu},\Omega\}=0
\end{equation} 
the product of the operators of the modified Dirac equations in curved spacetime yields their parent spin-1/2 KG equation (\ref{KG12}).\par Thus, by symmetrizing $ \nabla_{\alpha}X_{\beta} $, the Proca equation is replaced with a pure wave equation. Adding a complex scalar to the Dirac equations allowed the product of the  Dirac factorizations to yield the parent spin 1/2 wave equation. The KG equation (\ref{KG}) has been established as the proper covariant wave equation for spins 0,1,1/2 in curved spacetime. Moreover, the Lie derivative of the metric along the line element field connects MGR to the KG equation of quantum field theory for spins 1,0 and a pair of related spin-1/2 Dirac spinors. A natural unification of gravity and quantum field theory has been established. 
\subsection{Spin-2 Klein-Gordon equation}
The KG equation for the symmetric spin-2 field $ \tilde{\chi}_{\alpha\beta} $ is 
\begin{equation}\label{S2}
	\nabla_{\mu}\nabla^{\mu}\tilde{\chi}_{\alpha\beta}=k^{2}\tilde{\chi}_{\alpha\beta}.
\end{equation}$ \tilde{\chi}_{\alpha\beta} $ is constructed from the symmetrization of the general (0,2) tensor field $ \chi_{\alpha\beta} $ according to 
\begin{equation}\label{Sym2}
	\begin{split}
		\chi_{\alpha\beta}&=\frac{1}{2}(\chi_{\alpha\beta}+\chi_{\beta\alpha})+\frac{1}{2}(\chi_{\alpha\beta}-\chi_{\beta\alpha})\\
		&:=\frac{1}{2}\tilde{\chi}_{\alpha\beta}+\frac{1}{2}C_{\alpha\beta}.
	\end{split}
\end{equation} The symmetric spin-2 field must be divergenceless and traceless with respect to the metric; it has 5 degrees of freedom if $ k\neq0 $. Hence, $ \tilde{\chi}_{\alpha\beta} $ cannot be equivalent to $ \tilde{\varPsi}_{\alpha\beta} $ because the traceless attribute of the spin-2 field would lead to the trivial solution of $ \tilde{\varPsi}_{\alpha\beta} $. Gravitons are the particles associated with the metric. Hence, the symmetric spin-2 field must involve the metric as a field variable to enable gravitons to be described by the spin-2 KG equation. 
\par An expression for $ \tilde{\chi}_{\alpha\beta} $ can be obtained from the Orthogonal Decomposition Theorem in Appendix A, which involves a linear sum of (0,2) divergenceless tensors as one part of the orthogonal splitting of an arbitrary (0,2) symmetric tensor. That collection of divergenceless tensors can be defined to consist of the union of the set of Lovelock tensors and a set of symmetric divergenceless tensors, which are not Lovelock tensors. The collection of non-Lovelock tensors is represented by $ h_{\alpha\beta} $. They are independent of both the metric and the Einstein tensor, which are the Lovelock tensors in a four-dimensional spacetime. \par The general symmetric tensor $ \tilde{\chi}_{\alpha\beta} $ must be a linear combination of the matter energy-momentum tensor and an unknown symmetric tensor $ w_{\alpha\beta} $;
\begin{equation}\label{key}
	\tilde{\chi}_{\alpha\beta}=a\tilde{T}_{\alpha\beta}+bw_{\alpha\beta}
\end{equation}where $ a $ and $ b $ are arbitrary constants. $ w_{\alpha\beta} $ can then be orthogonally decomposed into $\varPhi_{\alpha\beta}  $, the Lovelock tensors and the non-Lovelock tensors:
\begin{equation}\label{key}
	\tilde{\chi}_{\alpha\beta}=a\tilde{T}_{\alpha\beta}+b(G_{\alpha\beta}+\Lambda g_{\alpha\beta}+\varPhi_{\alpha\beta}+h_{\alpha\beta}).
\end{equation} By requiring both $ \tilde{\chi}_{\alpha\beta} $ and $ h_{\alpha\beta} $ to be traceless, equation (\ref{MEQ}) is recovered for a non-degenerate inverse metric by setting $ a=-1 $ and $ b=\frac{c^{4}}{8\pi G} $. It follows that 
\begin{equation}\label{Psi2}
	\tilde{\chi}_{\alpha\beta}=\frac{c^{4}}{8\pi G}h_{\alpha\beta}.
\end{equation}$h_{\alpha\beta} $ has dimensions of $ L^{-2} $. That eliminates the super-energy divergenceless tensors such as the trace of the Chevreton tensor and the Bach tensor, which have dimension $ L^{-4} $. \par Thus, a general spin-2 symmetric tensor that satisfies the symmetric spin-2 KG equation contains the Einstein equation; it is hidden in the spin-2 field $\tilde{\chi}_{\alpha\beta} $. The metric itself explicitly belongs to the Einstein equation as one of the Lovelock tensors. The metric is stripped-out of $\tilde{\chi}_{\alpha\beta}  $ and is not available as a field variable in the spin-2 KG equation.
\par Some consequences of this result are now discussed.
\section{Spin-2 ``particles"}
Gravitons are considered to be massless particles because of the $\frac{1}{r^{2}}$ long-range effective force behaviour of gravity. They have spin-2 so that they can couple to the energy-momentum tensor. However, in a four-dimensional spacetime when the mass vanishes in (\ref{S2}), 
\begin{equation}\label{habwave}
	\nabla_{\mu}\nabla^{\mu}h_{\alpha\beta}=0
\end{equation} is the equation describing a massless spin-2 ``particle". Because $h_{\alpha\beta}$ is independent of the metric, this equation cannot describe a massless spin-2 graviton. Moreover, if the mass is nonzero, it follows from (\ref{S2}) and (\ref{Psi2}) that 
\begin{equation}\label{key}
	\nabla_{\mu}\nabla^{\mu}h_{\alpha\beta}=k^{2}h_{\alpha\beta}	
\end{equation}so neither massless nor massive gravitons exist.\par Furthermore, spin-2 ``particles" cannot couple to a non-zero energy-momentum tensor as force mediators for gravity. The action functional $ S^{H} $ that generates (\ref{S2}) from the variation with respect to $ h_{\alpha\beta} $ is
\begin{equation}\label{key}
	S^{H}=-\frac{1}{2}\nabla_{\mu}h_{\rho\sigma}\nabla^{\mu}h^{\rho\sigma}-\frac{1}{2}k^{2}h_{\rho\sigma}h^{\rho\sigma},
\end{equation} and the interaction of the matter energy-momentum tensor with $h_{\rho\sigma}$ is
\begin{equation}
	\begin{split}
		S_{h}^{int} &= -\frac{1}{2c} \int \tilde{T}^{\rho\sigma}h_{\rho\sigma} \sqrt{-g}d^{4}x
		\\&=-\frac{c^{3}}{16\pi G} \int (G^{\rho\sigma}+\varPhi^{\rho\sigma})h_{\rho\sigma}\sqrt{-g}d^{4}x.
	\end{split}
\end{equation} The variation  of $ S_{h}^{int}$ with respect to $ h_{\alpha\beta} $ must vanish; otherwise $ S^{H}+S_{h}^{int}$ would not generate the spin-2 wave equation (\ref{S2}). Since both $\varPhi^{\rho\sigma}$ and $G^{\rho\sigma}  $ are independent of $ h_{\rho\sigma} $, the variation of $ S_{h}^{int}$ with respect to $ h_{\alpha\beta} $ requires $ G^{\alpha\beta}+\varPhi^{\alpha\beta}=0 $. It follows that $ \tilde{T}^{\alpha\beta} $ vanishes and there is no coupling to the matter energy-momentum tensor. Spin-2 ``particles" do not couple to any type of matter but can occupy the vacuum in accordance with $G^{\alpha\beta}+\varPhi^{\alpha\beta}=0 $. Gravity, unlike the other three known forces in nature, does not require the exchange of particles to describe its long-range force behavior.   
\subsection{The hierarchy problem}
The hierarchy problem of particle physics can be stated as the question: why is the force of gravity much weaker than the other three known forces in nature?  In the case of electrodynamics, if both gravity and electrodynamics have long-range massless force mediators, why is the electromagnetic force $ 10^{40} $ times stronger than that of gravity?  The electroweak force is $ 10^{24} $ times stronger than gravity. Moreover, as the name suggests, the strong nuclear force presents the largest disparity to gravity at nuclear dimensions. \par At the root of this problem is the notion that gravity \emph{must} be quantized. However, the symmetric spin-2 KG equation in a 4-dimensional Lorentzian spacetime excludes gravitons as force mediators of gravity. That starkly contrasts with the spin-1 KG equation for a massless photon, which mediates the electromagnetic field. Similarly, the massive spin-1 W and Z bosons mediate the electroweak force and the spin-1 massless gluons mediate the strong nuclear force. Gravity has neither massless nor massive particles that act as force mediators. Thus, the hierarchy problem is explained without the need of extra spatial dimensions inherent in string theory. 
\section{The spin-1 field and gravity}
The symmetrization of the asymmetric tensor $\varPsi_{\alpha\beta}  $ into the Lie derivative of the metric and the Faraday tensor for electrodynamics when $ k=0 $, is somewhat similar to the approach that Einstein presented \cite{31,32,33} to unify gravity and the electromagnetic field. He generalized the Riemannian metric as an asymmetric tensor with the symmetric part representing the gravitational field in a Riemannian spacetime, and the antisymmetric component with six remaining degrees of freedom describing the electromagnetic field. However, one major problem with his theory was that it did not encompass quantum theory; whereas $\varPsi_{\alpha\beta}  $ comes from quantum theory and while $\tilde{\varPsi}_{\alpha\beta}  $ is not the metric, it is the Lie derivative of the metric along the quantum vector.\par From the structure of the metric in equation (\ref{gab}), the spin-1 field is an integral part of the Lorentzian metric and therefore the gravitational field. That follows from the scale invariance of the spin-1 KG equation. Consider 
\begin{equation}\label{KGl}
	\nabla_{\alpha}\nabla^{\alpha}(\lambda X^{\beta})=k^{2}\lambda X^{\beta},
\end{equation}where $ \lambda>0 $ is an arbitrary scalar, $ X^{\beta}=fu^{\beta} $, $ X_{\beta}=fu_{\beta}  $ and $ f\neq0 $. Provided $ \Box \lambda=-\frac{2}{f}\partial_{\alpha}f\partial^{\alpha}\lambda$, (\ref{KGl}) is independent of $ \lambda $. Choosing $ \lambda=\frac{1}{f} $ requires 
\begin{equation}\label{f0}
	\Box f=0  
\end{equation}
and the unit vectors collinear with the line element field satisfy the spin-1 KG equation
\begin{equation}\label{KG1u}
	\nabla_{\alpha}\nabla^{\alpha}u_{\beta}=k^{2}u_{\beta}.
\end{equation}The constraint (\ref{f0}) is compatible with (\ref{f}) because it is possible to choose $ f_{(\kappa)}=\frac{C-f}{2}$ where $ C $ is an arbitrary constant.
Thus, gravity and the spin-1 fields are automatically unified in a Lorentzian spacetime if the magnitude of the spin-1 fields satisfies the homogeneous wave equation (\ref{f0}). The spin-1 fields contribute to the gravitational field in a Lorentzian spacetime.
\section{Quantum entanglement, determinism and the wave-particle duality}
The outer product $\varPsi^{a}\otimes\varPsi^{a\dagger} $ in (\ref{KGC12}) is fundamental to the existence of entangled spin-1/2 fermion states. Each of the four 2$\times$2 blocks of $ \varPsi^{a}_{\dot{a}} $ can be expanded as a sum of products in terms of basis spinors. For example, if $ \varphi^{A}_{i} $ and  $ \varphi^{\dot{A}}_{j} $ are basis spinors for $ \varPsi^{A} $ and $ \varPsi^{\dot{A}} $ in the block $ \varPsi^{A\dot{A}} $, then $ \varPsi^{A}=\displaystyle\sum_{i=1}^{2}a_{i}^{A}\varphi^{A}_{i} $, $ \varPsi^{\dot{A}}= \displaystyle\sum_{j=1}^{2}b_{j}^{\dot{A}}\varphi^{\dot{A}}_{j} $ and $ \varPsi^{A\dot{A}}=\displaystyle\sum_{i,j=1}^{2}a^{A\dot{A}}_{ij}\varphi^{A}_{i}\otimes\varphi^{\dot{A}}_{j} $. Any term in the sum with coefficients \emph{not} of the form $a_{ij}^{A\dot{A}}=a_{i}^{A}b_{j}^{\dot{A}} $ is a mixed entangled state. Thus, $\varPsi^{a}_{\dot{a}}  $ can generate a pair of entangled spinors. However, how is quantum information exchanged between the two spin-1/2 fermions?\par The symmetric tensor $ \tilde{\varPsi}_{\alpha\beta}$ is the Lie derivative of the metric. It is the symmetric part of the spin-1 KG equation and the symmetric part of $\varPsi^{a}_{\dot{a}}  $. Given a diffeomorphism $ \phi: M\longrightarrow M $, $ \tilde{\varPsi}_{\alpha\beta}$ is constructed from the pullback $\phi_{t\ast}$ of the metric under the diffeomorphism group, Diff(M):
\begin{equation}\label{key}
	\tilde{\varPsi}_{\alpha\beta}= \lim_{t \to 0} \{\frac{\phi_{t\star}[g_{\alpha\beta}(\phi_{t}(p))]-g_{\alpha\beta}(p)}{t}\}
\end{equation}where $\phi_{t}  $ is the flow down the integral curves associated with the line element vector $ X^{\beta}=\dfrac{dx^{\beta}}{dt} $ and $ t $ defines a one-parameter family of diffeomorphisms. The Lorentz group is a subgroup of Diff(M) so the pullback of the metric is not restricted to the Lorentz group. The metric at a point $ p^{\prime} $ on the integral curve associated with the regular vector $ X^{\beta} $, far from a given point $ p $ on that curve, can be pulled back instantaneously to the neighbourhood of $ p $, or pushed forward from $ p $ with $(\phi_{t}^{-1})^{\ast}  $. Hence, the rate of change of the gravitational field along $ X^{\beta} $ is not limited to the speed of light, and could be instantaneous with the caveat that there may be a presently unknown upper bound to the speed that spacetime itself can transfer information. That does not conflict with the Lorentzian behaviour of gravitational waves, which carry energy; information is not energy and has no mass equivalent. Information from the  measurement of the spin of the distant fermion would be instantaneously transmitted to the fermion at point $p$ by the Lie derivative of the metric along the line element vector. The two fermions would have opposite spins due to the Pauli exclusion principle and the entanglement would be destroyed by the measurement. Although bosons obey different quantum statistics, a similar argument describes their entanglement.  \par That establishes the global nature of quantum theory and the concept of entanglement within a Lorentzian invariant formalism. The respective local/nonlocal characteristics of MGR and quantum field theory no longer present an insurmountable problem to unify the theories. \par  The KG equation for a spin-1 boson, and the KG equation for the outer product of a spin-1/2 fermion with its Hermitian conjugate, have antisymmetric and symmetric components. The antisymmetric tensor $ K_{\alpha\beta} $  represents the boson or a pair of spin-1/2 particles and satisfies
\begin{equation}\label{key}
	\nabla_{\mu}K_{\alpha\beta}+\nabla_{\beta}K_{\mu\alpha}+\nabla_{\alpha}K_{\beta\mu}=0
\end{equation}
from which the wave equation 
\begin{equation}\label{Km}
	\begin{split}
		\nabla_{\mu}\nabla^{\mu}K_{\alpha\beta}=-2k^{2}K_{\alpha\beta}-2K{^{\mu}}_{[\alpha}R_{\beta]\mu}-2K^{\mu\sigma}R_{\mu\alpha\sigma\beta}-2\nabla_{[\alpha}\nabla^{\mu}\tilde{\varPsi}_{\beta]\mu}
	\end{split}
\end{equation}is obtained using (\ref{P2}). From (\ref{ME}) and (\ref{GRKG})
\begin{equation}\label{PsiT}
	\frac{1}{2}\tilde{\varPsi}_{\alpha\beta}=\frac{8\pi G}{c^{4}}\tilde{T}_{\alpha\beta}-G_{\alpha\beta}-u^{\lambda}(u_{\alpha}\nabla_{\beta}X_{\lambda}+u_{\beta}\nabla_{\alpha}X_{\lambda}).
\end{equation}  Given a particular metric, solutions to (\ref{PsiT}) exist in a region of spacetime outside of matter according to $ \varPhi_{\alpha\beta}+G_{\alpha\beta}=0 $ as shown in \cite{18} and Appendix B. The components of the quantum vector are then apparent in terms of the metric variables, and equation (\ref{Km}) will be satisfied. It is well known \cite{17,25} that the Einstein equation and wave equations are deterministic so it is conjectured that equations (\ref{Km}) and (\ref{PsiT}) are deterministic.  \par 
The last term in (\ref{Km}), $-2\nabla_{[\alpha}\nabla^{\mu}\tilde{\varPsi}_{\beta]\mu}:=N_{\alpha\beta} $, is constructed from the Lie derivative of the metric along the quantum vector. That term involves three derivatives of the line element field. It is shown in (\ref{N133}) of appendix C from the solutions to (\ref{ME}) in a region outside of matter with the metric of appendix B, that $ N_{13}$ varies as $ \frac{1}{fr^{4}} $ where $ f\neq0 $ is the magnitude of $ X_{3} $. $N_{13} $ is sensitive to changes in the gravitational field at minute distances. In the microworld, it can guide the wave in response to local changes due to gravity. The gravitationally guided wave is similar in concept to the pilot wave of the de Broglie-Bohm theory \cite{20,21}. That theory is inherently nonlocal and deterministic. These characteristics are shared in (\ref{Km}), which depends on $ \tilde{\varPsi}_{\alpha\beta}$ and the wave behavior of $ K_{\alpha\beta} $ that represents the particle. However, the gravitationally guided wave does not suffer from the mystery of the quantum potential \cite{34} integral to pilot wave theory.  
\par In contrast to a pair of related spin-1/2 fermions, a single spin-1/2 fermion is described by the wave equation (\ref{KG12}), which does not involve the Lie derivative of the metric. The spin-1/2 fermion travels as a wave at all spacetime coordinates, but does not have the gravitational sensitivity that bosons or a pair of related spin-1/2 fermions intrinsically possess.
\par   When we undertake a measurement of a particle in the microworld by bombarding it with light quanta or other particles, chaos or destruction is inherent in the process of measurement. Despite that the measurement may constrain the wave behaviour of the particle, it does not permanently destroy that property of the particle when detected as a particle, and it does not suddenly become a wave when detected as a wave, because in both cases, the particle behaves as a wave at all spacetime coordinates according to the fundamental wave equation (\ref{KG}). The experiment can be designed to detect a wave or a particle, but the reality of the microworld is not determined by the experiment; reality exists in the quantum-world before and after a measurement is performed on an entity in it. That is contrary to the Copenhagen interpretation \cite{35} (and references therein) of quantum theory whereby the reality of the microworld is declared after a measurement is performed on a particle in it. The physical process of measuring a particle in the microworld requires the statistical approach of quantum theory to describe the chaotic event. However, as discussed in \cite{36}, a particle moves deterministically as a wave until a measurement is performed on the system or another chaotic event happens that involves the particle. Then, the particle "jumps" from one quantum state to another with a different probability. It then continues along a new deterministic trajectory. \par The complementarity principle \cite{37,38} mysteriously provides entities in the microworld with wave or particle characteristics. The wave-particle duality does not challenge the reality of the microworld. The duality of the microworld melds into the macroworld by allowing the experimenter to determine the wave or particle characteristics as designed in the experiment. However, both properties cannot be measured simultaneously due to the Heisenberg uncertainty principle. \par Quantum particles fundamentally move as a  wave in accordance with the KG equation (\ref{KG}). That the particle intrinsically moves as a wave removes the mystery in the complementarity principle. However, there remains the problem of simultaneously observing the wave-particle properties. Entanglement, another mystery of quantum theory, solves that problem. In the 2017 study \cite{39}, the deterministic wave-particle entanglement of two photons was achieved. More importantly, a single self-entangled photon has recently been observed to exhibit simultaneous wave and particle behaviours \cite{40}. That was the first experiment at the single particle level, which is required to test a quantum-mechanical entity acting as both a particle and a wave. It is also interesting that the surface plasmon polariton experiment \cite{41} observed the simultaneous behavior of light acting as both a wave and a stream of particles. 

\section{Conclusion}
The results in this article are obtained directly from the Lie derivative of a Lorentzian metric along a line element vector field. Every four-dimensional non-compact paracompact Hausdorff Lorentzian manifold $(M,g_{\alpha\beta}) $  admits a line element field, ($X^{\beta},-X^{\beta}$). MGR exploits that vector field with the introduction of a tensor that describes the energy-momentum of the gravitational field and completes General Relativity. That tensor is constructed from the Lie derivatives of the metric and a product of unit line element covectors along the line element vector. \par Particles with  spins 0,1,2,1/2 move as a wave at all spacetime coordinates in accordance with the multi-spin Klein-Gordon equation of quantum field theory in curved spacetime. The spin-1 KG equation is derived from the line element field and its covariant derivative; the spin-0 KG equation is constructed from the contraction of the latter with the metric and the KG equation of the outer product of a spin-1/2 Dirac spinor and its Hermitian
conjugate yields twice the spin-1 KG equation. Thus, the line element vector field plays the role of a fundamental quantum vector field.\par Symmetrization of the KG equation in curved spacetime introduced a fundamental geometrical entity into quantum field theory that has been neglected. The symmetric part of the KG equation for spins 0,1, and the outer product of a spin-1/2 Dirac spinor and its Hermitian conjugate, is the Lie derivative of the metric along the line element vector. Thus, the Lie derivative of the metric along a line element vector field is the natural link of Modified General Relativity to quantum field theory in curved spacetime. \par From the structure of a Lorentzian metric in terms of a Riemannian metric and a product of unit line element covectors, it is evident that gravity and the spin-1 fields are unified in a Lorentzian spacetime provided that the magnitude of the spin-1 fields satisfies a homogeneous wave equation.
\par A spin-2 decomposition for $ \tilde{\chi}_{\alpha\beta} $ is obtained in terms of a collection of tensor fields independent of the Lovelock tensors. It follows that the Einstein equation is intrinsically hidden in $ \tilde{\chi}_{\alpha\beta} $. The metric does not appear as a field variable in the spin-2 KG equation. Spin-2 ``particles" do not couple to a non-zero energy-momentum tensor as force mediators for gravity, but occupy the vacuum. Thus, gravitons in a Lorentzian four-dimensional spacetime described by the covariant wave equation (\ref{KG}) do not exist. Unlike the other three known fundamental forces in nature, no particle exchange is required to explain the effective force of gravity; that is nicely done by the curvature of spacetime. That neither massless nor massive gravitons exist explains the hierarchy problem of particle physics. 
\par  Bosons with spins 0,1,2, a spin-1/2 fermion and a Hermitian pair of spin-1/2 fermions move as a wave at all spacetime coordinates in the microworld. Reality exists in the microworld before and after a measurement of an entity in the microworld. The process of measurement does not destroy the wave characteristics of the particle, or the particle nature of the wave.  \par The wave equation for the spin-1 bosons and the Hermitian pair of spin-1/2 fermions contains a term that is constructed from the Lie derivative of the metric along the quantum vector. That term is sensitive to changes in the local gravitational field that guide the wave in response to local changes due to gravity.  \par $\tilde{\varPsi}_{\alpha\beta}$ is constructed from the pullback of the metric under the diffeomorphism group. Since that group is not restricted to the Lorentz group, information can be pulled back or pushed forward along the quantum vector at a superluminal speed. That establishes the nonlocal behaviour of quantum theory and the concept of entanglement within a Lorentzian formalism. Thus, by using the line element field within MGR and introducing a symmetric tensor into quantum theory by symmetrizing the KG equation in curved spacetime, the respective local/nonlocal characteristics of MGR and quantum field theory no longer present an insurmountable problem to unify the theories.
\section*{Acknowledgements}
I would like to thank the anonymous referee who provided constructive comments that improved the quality of the manuscript.
\appendix
\section{Orthogonal Decomposition Theorem}
\begin{theorem}
	Orthogonal Decomposition Theorem (ODT): An arbitrary (0,2) symmetric tensor $ w_{\alpha\beta} $ in the symmetric cotangent bundle $ S^{2}T^{\ast}M $ on an n-dimensional Lorentzian manifold $ (M,g_{\alpha\beta}) $ with a torsionless and metric compatible connection can be orthogonally decomposed as
	\begin{equation}\label{ODT}
		w_{\alpha\beta}= v_{\alpha\beta}+ \varPhi_{\alpha\beta} 
	\end{equation} where $v_{\alpha\beta}  $ represents a linear sum of symmetric divergenceless (0,2) tensors and $\varPhi_{\alpha\beta}=\frac{1}{2}\pounds_{X}g_{\alpha\beta}+\pounds_{X}u_{\alpha}u_{\beta}$ where the unit vector $\bm{u}  $ is collinear with one of the pair of regular vectors in the  line element field $ (\bm{X},-\bm{X}) $.
\end{theorem}

\begin{proof}
	Let the Lorentzian manifold $ (M,g_{\alpha\beta}) $ be non-compact paracompact and Hausdorff. A smooth regular line element field $(\bm{X},\bm{-X)}$ exists as does a unit vector $ \bm{u} $ collinear with one of the pair of line element vectors. Let M be endowed with a smooth Riemannian metric $ g^{+}_{\alpha\beta} $. The smooth Lorentzian metric $ g_{\alpha\beta} $ is constructed \cite{16,17,24,25} from $ g^{+}_{\alpha\beta} $ and the unit covectors $ u_{\alpha}$ and $ u_{\beta}$ : $g_{\alpha\beta}=g^{+}_{\alpha\beta}-2u_{\alpha}u_{\beta} $. Let $ w_{\alpha\beta} $ and $ v_{\alpha\beta} $ belong to $ S^{2}T^{\ast}M $, the cotangent bundle of symmetric $(0,2)$ tensors on M. In the compact neighborhood of a point $p$ in an open subset of $ S^{2}T^{\ast}M $ which contains $ g^{+}_{\alpha\beta} $, an arbitrary $ (0,2) $ symmetric tensor $ w_{\alpha\beta} $ can be orthogonally and uniquely decomposed by the Berger-Ebin theorem \cite{42} according to \begin{equation}\label{key}
		w_{\alpha\beta}=v_{\alpha\beta}+\frac{1}{2}\pounds_{\xi}g^{+}_{\alpha\beta}
	\end{equation} where $ \bm{\xi} $ is an arbitrary vector and $v_{\alpha\beta}  $ represents a linear sum of symmetric divergenceless (0,2) tensors: $ {\nabla^{+}}^{\alpha}v_{\alpha\beta}=0 $.\par The divergence of $ v_{\beta}^{\alpha} $ in the mixed tensor bundle can be written as $ \nabla_{\alpha}v_{\beta}^{\alpha}=\partial_{\alpha}v_{\beta}^{\alpha}+\frac{v_{\beta}^{\lambda}}{2g}\partial_{\lambda}g-\frac{1}{2}v^{\alpha\lambda}\partial_{\beta}g_{\alpha\lambda} $. Since the determinant of $ g_{\alpha\beta} $, $ g $, is related to that of $ g^{+}_{\alpha\beta} $ by $ g=-g^{+} $ 
	\begin{equation}\label{Dv}
		\begin{split}
			\nabla_{\alpha}v_{\beta}^{\alpha}-\nabla^{+}_{\alpha}v_{\beta}^{\alpha}=v^{\alpha\lambda}\partial_{\beta}(u_{\alpha}u_{\lambda}).
		\end{split}
	\end{equation} The left hand side of (\ref{Dv}) is a (0,1) tensor but the right hand side is not, which demands:
	\begin{equation}\label{}
		v^{\alpha\lambda}\partial_{\beta}(u_{\alpha}u_{\lambda})=0
	\end{equation}where $\partial_{\beta}u_{\alpha}\neq0 $. That guarantees $\nabla^{\alpha}v_{\alpha\beta}=0  $ because $\nabla^{+\alpha}v_{\alpha\beta}=0  $. Hence, 
	\begin{equation}\label{decomp}
		%	\begin{aligned}
		w_{\alpha\beta}=v_{\alpha\beta}+\frac{1}{2}\pounds_{\xi}g_{\alpha\beta}+\pounds_{\xi}u_{\alpha}u_{\beta}
		%	\\&=v_{\alpha\beta}+\varPhi_{\alpha\beta}
		%	\end{aligned}
	\end{equation} where $ \nabla^{\alpha}v_{\alpha\beta}=0 $. $ \xi^{\lambda} $ is an arbitrary vector which can be chosen to be collinear to $ u^{\lambda} $. Without loss of generality, $ \xi^{\lambda} $ can then be replaced by $ X^{\lambda} $. Using $ X^{\lambda}=fu^{\lambda} $ where $f\neq0  $ is the magnitude of $X^{\lambda}  $, the expression  $X^{\lambda}\nabla_{\lambda}(u_{\alpha}u_{\beta}) $ in the last term of (\ref{decomp}) then vanishes in an affine parameterization and
	\begin{equation}\label{udecomp}
		w_{\alpha\beta}=v_{\alpha\beta}+\varPhi_{\alpha\beta}
	\end{equation}where
	\begin{equation}\label{Phiab}
		\varPhi_{\alpha\beta}:=\frac{1}{2}(\nabla_{\alpha}X_{\beta}+\nabla_{\beta}X_{\alpha})+u^{\lambda}(u_{\alpha}\nabla_{\beta}X_{\lambda}+u_{\beta}\nabla_{\alpha}X_{\lambda}).
	\end{equation} The decomposition is orthogonal: $<v_{\alpha\beta},\varPhi_{\alpha\beta}>=0  $. \\
	%	\qed
\end{proof}
\section{Spherical solution of MGR in the equatorial plane}
In a region of spacetime where there is no matter, $ \tilde{T}_{\alpha\beta}=0 $ and the field equations must satisfy
\begin{equation}\label{EF}
	G_{\alpha\beta}+\varPhi_{\alpha\beta}=0.
\end{equation} Spherical solutions to these nonlinear equations are now investigated in the equatorial plane. The metric has the form
\begin{equation}\label{g}
	ds^{2}=-e^\nu c^{2}dt^{2}+e^{\lambda}dr^{2}+r^{2}(d\theta^{2}+sin^{2}\theta d\varphi^{2})
\end{equation} where $ \nu $ and $ \lambda $ are functions of t, r and $\varphi$ with a fixed $ \theta $ and vanishing derivatives with respect to $ \theta $. The non-zero connection coefficients (Christoffel symbols) are:
\begin{flalign*}
	&\Gamma^{0}_{00}=\frac{1}{2}\partial_{0}\nu,\Gamma^{0}_{01}=\frac{1}{2}\partial_{1}\nu,\Gamma^{0}_{03}=\frac{1}{2}\partial_{3}\nu,
	\Gamma^{0}_{11}=\frac{1}{2}\partial_{0}\lambda e^{\lambda-\nu},&\\
	&\Gamma^{1}_{00}=\frac{1}{2}\partial_{1}\nu e^{\nu-\lambda},\Gamma^{1}_{01}=\frac{1}{2}\partial_{0}\lambda, 
	\Gamma^{1}_{11}=\frac{1}{2}\partial_{1}\lambda, \Gamma^{1}_{13}=\frac{1}{2}\partial_{3}\lambda,\Gamma^{1}_{22}=-re^{-\lambda}, \Gamma^{1}_{33}=-r\sin^{2}\theta e^{-\lambda},&\\
	&\Gamma^{2}_{12}=\frac{1}{r},\Gamma^{2}_{33}=-\sin\theta\cos\theta,\Gamma^{3}_{00}=\frac{e^{\nu}\partial_{3}\nu}{2r^{2}\sin^{2}\theta}, \Gamma^{3}_{11}=-\frac{e^{\lambda}}{2r^{2}\sin^{2}\theta}\partial_{3}\lambda,\Gamma^{3}_{13}=\frac{1}{r},\Gamma^{3}_{23}=\cot\theta.&
\end{flalign*}
The unit vectors $ u^{\alpha} $ satisfy 
\begin{equation}\label{uu1}
	u^{\alpha}u_{\alpha}=-1.
\end{equation} To study this highly nonlinear set of equations given by (\ref{EF}) with the property (\ref{uu1}), $ u_{2} $ is chosen to vanish. This requires
\begin{equation}\label{X2}
	X_{2}=0 
\end{equation} 
because $ u_{\alpha} $ is collinear with $ X_{\alpha}$.  Static solutions to (\ref{EF}) are sought which require the components of the line element field to satisfy 
\begin{equation}\label{xcomp}
	\partial_{0}X_{\alpha}=0,
\end{equation} and from the metric,
\begin{equation}\label{dlambdanu}
	\partial_{0}\lambda=0,\enspace \partial_{0}\nu=0. 
\end{equation}From (\ref{umu}), $ u_{0}=0 $ and therefore $ X_{0}=0 $. $X_{1} $ and $ X_{3} $ are non-zero.\par The components of $ \varPhi_{\alpha\beta} $ to be considered are then: 
\begin{equation}\label{key}
	\varPhi_{00}=-\frac{1}{2}e^{\nu-\lambda}\nu^{\prime} X_{1}-\frac{e^{\nu}}{2r^{2}\sin^{2}\theta}\partial_{3}\nu X_{3},
\end{equation}
\begin{equation}\label{key}
	\varPhi_{11}=(1+2u_{1}u^{1} )({X_{1}}^{\prime}-\frac{1}{2}\lambda^{\prime}X_{1}+\frac{e^{\lambda}}{2r^{2}\sin^{2}\theta}\partial_{3}\lambda X_{3}),
\end{equation} 
\begin{equation}\label{key}
	\varPhi_{22}=re^{-\lambda} X_{1},  
\end{equation} 
\begin{equation}\label{key}
	\varPhi_{33}=(1+2u_{3}u^{3})(\partial_{3}X_{3}+r \sin^{2}\theta e^{-\lambda}X_{1}),
\end{equation} the Ricci scalar, which from (\ref{EF}) equals $ \Phi $, is
\begin{equation}\label{R}
	\begin{split}
		R=e^{-\lambda}(-\nu^{\prime\prime}-\frac{1}{2}{\nu^{\prime}}^{2}+\frac{1}{2}\lambda^{\prime}\nu^{\prime}-\frac{2}{r}\nu^{\prime}+\frac{2}{r}\lambda^{\prime}-\frac{2}{r^{2}})+\frac{2}{r^{2}}-\frac{1}{2r^{2}\sin^{2}\theta}[2\partial_{3}\partial_{3}\nu+2\partial_{3}\partial_{3}\lambda+(\partial_{3}\lambda)^{2}\\+(\partial_{3}\nu)^{2}+\partial_{3}\lambda\partial_{3}\nu],
	\end{split}
\end{equation}
and the corresponding components of the Einstein tensor are:
\begin{equation}\label{key}
	\begin{split}
		G_{00}=\frac{1}{r^{2}}e^{\nu-\lambda}(r\lambda^{\prime}-1+e^{\lambda})+\frac{e^{\nu}}{2r^{2}\sin^{2}\theta}[-\partial_{3}\partial_{3}\lambda-\frac{1}{2} (\partial_{3}\lambda)^{2}],
	\end{split}
\end{equation}
\begin{equation}\label{key}
	\begin{split}
		G_{11}=\frac{1}{r^{2}}(1+r\nu^{\prime}-e^{\lambda})+\frac{e^{\lambda}}{2r^{2}\sin^{2}\theta}[\partial_{3}\partial_{3}\nu+\frac{1}{2}({\partial_{3}\nu})^{2}],
	\end{split}
\end{equation}
\begin{equation}\label{key}
	\begin{split}
		G_{22}=\frac{r^{2}e^{-\lambda}}{2}[\nu^{\prime\prime}+\frac{1}{r}\nu^{\prime}-\frac{1}{r}\lambda^{\prime}-\frac{1}{2}\lambda^{\prime}\nu^{\prime}+\frac{1}{2}(\nu^{\prime})^{2}]+\frac{1}{2\sin^{2}\theta}[\partial_{3}\partial_{3}\nu+\partial_{3}\partial_{3}\lambda+\frac{1}{2}({\partial_{3}\lambda})^{2} \\+\frac{1}{2}({\partial_{3}\nu})^{2}+\frac{1}{2}\partial_{3}\nu\partial_{3}\lambda],
	\end{split}
\end{equation}
\begin{equation}\label{key}
	\begin{split}
		G_{33}=\frac{r^{2}\sin^{2}\theta e^{-\lambda}}{2}[-\frac{\lambda^{\prime}}{r}+\frac{\nu^{\prime}}{r}+\nu^{\prime \prime}+\frac{1}{2}{(\nu^{\prime}})^{2}-\frac{1}{2}\lambda^{\prime}\nu^{\prime}]+\frac{1}{2}[\partial_{3}\partial_{3}\nu+\partial_{3}\partial_{3}\lambda+\frac{1}{2}({\partial_{3}\lambda})^{2} \\+\frac{1}{2}({\partial_{3}\nu})^{2}+\frac{1}{2}\partial_{3}\nu\partial_{3}\lambda]
	\end{split}
\end{equation} where the prime denotes $ \partial_{1} $. \par These equations are greatly simplified by setting
\begin{equation}\label{lambdanu}
	\nu=-\lambda.
\end{equation}
Thus, a class of static spherical solutions to (\ref{EF}) are sought with the restrictions (\ref{X2}),(\ref{xcomp}),(\ref{dlambdanu}) and (\ref{lambdanu}). \par Since $ e^{\lambda-\nu}(\varPhi_{00}+G_{00})+\varPhi_{11}+G_{11}=0 $,
\begin{equation}\label{phiG0011}
	\begin{split}
		-\frac{e^{\lambda}}{r^{2}\sin^{2}\theta}\partial_{3}\partial_{3}\lambda+\frac{\lambda^{\prime}}{2}X_{1}+\frac{e^{\lambda}}{2r^{2}\sin^{2}\theta}\partial_{3}\lambda X_{3}\\+(1+2u_{1}u^{1})(X_{1}^{\prime}-\frac{\lambda^{\prime}}{2}X_{1}+\frac{e^{\lambda}}{2r^{2}\sin^{2}\theta}\partial_{3}\lambda X_{3})=0.
	\end{split}
\end{equation} which can be written as 
\begin{equation}\label{C1}
	-X_{1}^{\prime}+\lambda^{\prime}X_{1}+u_{3}u^{3}(\lambda^{\prime}X_{1}-2X_{1}^{\prime})+\frac{e^{\lambda}}{r^{2}\sin^{2}\theta}(-\partial_{3}\partial_{3}\lambda-u_{3}u^{3}X_{3}\partial_{3}\lambda)=0	
\end{equation}using $ u_{1}u^{1}=-u_{3}u^{3}-1 $.

\par  $G_{33}+\varPhi_{33}-\sin^{2}\theta (G_{22}+\varPhi_{22})=0 $ in the interval $ 0<\theta<\pi $ gives 
\begin{equation}\label{C2}
	(1+2u_{3}u^{3})\frac{\partial_{3}X_{3}}{\sin^{2}\theta}=-2re^{-\lambda}X_{1}u_{3}u^{3}.
\end{equation}
From $G_{22}+\varPhi_{22}=0  $
\begin{equation}\label{DM}
	\begin{split}
		-\lambda^{\prime\prime}+\lambda^{\prime 2}-\frac{2}{r}\lambda^{\prime}+\frac{2e^{\lambda}}{r^{2}}(X_{1}re^{-\lambda}+\frac{1}{4\sin^{2}\theta}(\partial_{3}\lambda)^{2})=0.
	\end{split}
\end{equation}A radially symmetric solution to (\ref{DM}) is sought with $ \partial_{3}\lambda=0 $ and $ \partial_{3}\partial_{3}\lambda\neq0 $ using the ansatz
\begin{equation}\label{X1ser}
	X_{1}=e^{\lambda}(a_{0}-\frac{b}{r}) 	
\end{equation}   where $ a_{0} $ and $ b $ are arbitrary parameters. Equation (\ref{DM}) then has the exact solution 
\begin{equation}\label{wow}
	\begin{split}
		\lambda(r)=-\ln(\frac{c_{1}}{r}+c_{2}-a_{0}r+2bln(r)),\enspace 0<r<\infty
	\end{split}
\end{equation}
where $c_{1}$ and $ c_{2}  $ are arbitrary constants. Equation (\ref{wow}) represents the extended Schwarzschild solution with $ c_{1}=\frac{-2GM}{c^{2}} $ and $ c_{2}=1 $.
\section{Radial dependence of $ N_{\alpha\beta} $}The last term in (\ref{Km}) can be written as
\begin{multline}\label{Nab}
	N_{\alpha\beta}=\frac{1}{2}(-g)^{-\frac{3}{2}}[\partial_{\beta}g\partial_{\mu}(\sqrt{-g}\tilde{\varPsi}^{\mu}_{\alpha})-\partial_{\alpha}g\partial_{\mu}(\sqrt{-g}\tilde{\varPsi}^{\mu}_{\beta})]+\frac{1}{\sqrt{-g}}[\partial_{\beta}\partial_{\mu}(\sqrt{-g}\tilde{\varPsi}^{\mu}_{\alpha})\\-\partial_{\alpha}\partial_{\mu}(\sqrt{-g}\tilde{\varPsi}^{\mu}_{\beta})]+\frac{1}{2}[\partial_{\beta}g_{\mu\lambda}\partial_{\alpha}\tilde{\varPsi}^{\mu\lambda}-\partial_{\alpha}g_{\mu\lambda}\partial_{\beta}\tilde{\varPsi}^{\mu\lambda}].
\end{multline}In the metric of Appendix A:\\
$ g^{00}=-e^{\lambda} $, $g^{11}=e^{-\lambda}$, $g^{22}=\frac{1}{r^{2}}$, $ g^{33}=\frac{1}{r^{2}\sin^{2}\theta} $, $ \sqrt{-g}=r^{2}\sin\theta$, $\tilde{\varPsi}^{0}_{0}=-e^{-\lambda}\lambda^{\prime}X_{1}-\frac{\partial_{3}\lambda}{r^{2}\sin^{2}\theta} $,\\ $\tilde{\varPsi}^{1}_{1}=e^{-\lambda}
(2X_{1}^{\prime}-\lambda^{\prime}X_{1})+\frac{X_{3}\partial_{3}\lambda}{r^{2}\sin^{2}\theta}  $, $\tilde{\varPsi}^{2}_{2}=\frac{2e^{-\lambda}X_{1}}{r}  $, $ \tilde{\varPsi}^{3}_{3}=\frac{2e^{-\lambda}X_{1}}{r}+\frac{2\partial_{3}X_{3}}{r^{2}\sin^{2}\theta}   $,\\ $\tilde{\varPsi}^{1}_{3}=e^{-\lambda}(\partial_{1}X_{3}-\frac{2X_{3}}{r}) $, $\tilde{\varPsi}^{3}_{1}=\frac{\partial_{1}X_{3}}{r^{2}\sin^{2}\theta}-\frac{2X_{3}}{r^{3}\sin^{2}\theta} $\\
$\tilde{\varPsi}^{00}=\lambda^{\prime}X_{1}+\frac{e^{\lambda}\partial_{3}\lambda}{r^{2}\sin^{2}\theta} $, $ \tilde{\varPsi}^{11}=e^{-2\lambda}(2X_{1}^{\prime}-\lambda^{\prime}X_{1})+\frac{e^{-\lambda}X_{3}\partial_{3}\lambda}{r^{2}\sin^{2}\theta} $, $\tilde{\varPsi}^{22}=\frac{2e^{-\lambda}X_{1}}{r^{3}}  $, $\tilde{\varPsi}^{33}=\frac{2e^{-\lambda}X_{1}}{r^{3}\sin^{2}\theta}+\frac{2\partial_{3}X_{3}}{r^{4}\sin^{4}\theta}  $.\\
\par The radial dependence of $ N_{13} $ is of particular interest. In the equatorial plane with $ \partial_{3}\lambda=0 $
\begin{equation}\label{N13}
	\begin{split}
		N_{13}=-\frac{4\partial_{3}\partial_{3}X_{3}}{r^{3}}+\frac{\partial_{1}\partial_{3}\partial_{3}X_{3}}{r^{2}}+\frac{e^{-\lambda}}{r^{3}}D
	\end{split}
\end{equation}where $ D $ represents a combination of dimensionless terms involving $ X_{3} $; its first and second derivatives; $ r $ and $ \lambda^{\prime} $.
From (\ref{X1ser}), equation (\ref{C2}) is equivalent to
\begin{equation}\label{key}
	\partial_{3}X_{3}=\frac{-2re^{-\lambda}X_{1}X_{3}^{2}}{f^{2}r^{2}+2X_{3}^{2}}
\end{equation}using $ X_{3}=fu_{3} $ where $ f\neq0 $ is the magnitude of $ X_{3} $. Assuming $ f $ is independent of $ \varphi $, \begin{equation}\label{key}
	X_{3}=\frac{2B\varphi+c_{4}\pm\sqrt{(2B\varphi+c_{4})^{2}+8A}}{4}
\end{equation}where $ A:=f^{2}r^{2} $, $B:=-re^{-\lambda} X_{1} $ and $ c_{4} $ is an arbitrary parameter. Setting  $c_{4}=0$ gives $ \partial_{3}\partial_{3}X_{3}=\pm\frac{AB^{2}}{(B^{2}\varphi^{2}+2A)^{3/2}} \neq0 $. There are no higher derivatives of $ X_{3} $ in $ \varphi $ so setting $ \varphi=0 $ gives the angular independent expression $ \partial_{3}\partial_{3}X_{3}=\pm\frac{B^{2}}{2\sqrt{2}fr} $. For small $ r $, $ B\simeq b $ and (\ref{N13}) has the structure for small $ r $
\begin{equation}\label{N133}
	N_{13}\simeq\mp \frac{2b^{2}}{\sqrt{2}fr^{4}}\mp\frac{b^{2}}{2\sqrt{2}r^{4}f}(\frac{rf^{\prime}}{f}+1)+O(r^{-3})
\end{equation}where the prime denotes $ \partial_{1} $.

\end{document}